\documentclass[sigconf,nonacm]{acmart}
\usepackage{subfigure}
\usepackage{amssymb}
\usepackage{amsmath}
\usepackage{anyfontsize}
\usepackage{mathtools}
\newcommand*{\origrightarrow}{}

\renewcommand*{\textrightarrow}{\fontfamily{cmr}\selectfont\origrightarrow}

\usepackage{amsthm}
\usepackage{times}
\usepackage{soul}
\usepackage{algorithm}
\usepackage{algorithmic}
\usepackage{booktabs} 
\newtheorem{theorem}{Theorem}
\newtheorem{definition}{Definition}
\newtheorem{lemma}{Lemma}

\usepackage{latexsym} 

\def\BibTeX{{\rm B\kern-.05em{\sc i\kern-.025em b}\kern-.08emT\kern-.1667em\lower.7ex\hbox{E}\kern-.125emX}}
    
\begin{document}

%
\title{A Lightweight Algorithm to Uncover Deep Relationships in Data Tables}

%

\author{Jin Cao}
\affiliation{%
  \institution{Nokia Bell Labs}
  }
\email{jin.cao@nokia-bell-labs.com}

\author{Yibo Zhao}
\affiliation{%
  \institution{Indeed Inc.}
}
\email{yibozhao@indeed.com}

\author{Linjun Zhang}
\affiliation{%
 \institution{University of Pennsylvania}
}
\email{linjunz@wharton.upenn.edu}

\author{Jason Li}
\affiliation{%
 \institution{Academy for Information Technology}
 \city{Scotch Plains}
 \state{NJ}
}
\email{jli1@ucvts.org}
%
\renewcommand{\shortauthors}{J. Cao et al.}

%
\begin{abstract}
Many data we collect today are in tabular form, with rows as records and columns as attributes associated with each record. Understanding the structural relationship in tabular data can greatly facilitate the data science process. 
Traditionally, much of this relational information is
stored in table schema and maintained by its creators, usually domain experts. In this paper, we develop automated methods to uncover deep relationships in a single data table without expert or domain knowledge. Our method can decompose a data table into layers of smaller tables, revealing its deep structure. The key to our approach is
a computationally lightweight forward addition algorithm that we developed to recursively
extract the {\em functional dependencies} between table columns that is scalable to tables with many columns.
With our solution, data scientists will be provided with automatically generated, data-driven insights when exploring new data sets.

\end{abstract}

%
%


%

%

%
\maketitle

\newcommand{\cm}[1]{}
\section{Introduction}

Prior to machine learning activities, the data scientist has to develop an understanding of the data. For a relational database, the contextual and structural information are expressed in data schema,  which is often generated by the domain experts at the time of data creation. However, schema maintenance 
on big or evolving data is non-trivial,
and often requires lots of manual effort and domain knowledge. 
In this paper, 
we aim to develop automated  profiling methods for tabular data that allow for a fast and accurate understanding of its structural relationships between data columns. Our work is motivated by {\em Automated Machine Learning} \cite{Matthias2015}, with the goal to develop automated procedures 
to improve efficiency of machine learning for non-experts.

\cm{
Track changes is off
Everyone
You
Guests
zlj11112222: N is defined in Section 2. Shall we define N again here?
Feb 18, 2019 8:06 PM

Hit Enter to reply
zlj11112222: Shall we briefly mention the meaning of $C\rightarrow_{\epsilon} Y$?
Feb 18, 2019 8:07 PM

Hit Enter to reply
Current file
Overview
16
 device + time are ID, then you can plot the other metric against time, aggregation w.r.t to the ID fields, slice and dice the data to form more compact summaries
 }

\cm{
Automated Machine Learning (AutoML) has become a popular topic over the past years \cite{Matthias2015}. Conventional machine learning crucially relies on domain experts, and it requires great human efforts to perform data exploration, problem formulation, feature extraction, hyper-parameter tuning, and model deployment. As the rapid growth of machine learning applications and the complexity of these tasks is often beyond non-experts, people would like an automated machine learning procedure that can be easily applied without experts knowledge. 
}

\cm{
Converting massive amount of collected data into information and actionable insights often takes a long time. 
This is because 
conventional machine learning crucially relies on domain experts to perform a range of machine learning tasks such as 
data exploration, problem formulation, feature extraction, hyper-parameter tuning and model deployment. As the rapid growth of machine learning applications and the complexity of these tasks is often beyond non-experts, there is an increasing demand for developing automated machine learning procedures that can be easily applied without experts knowledge. 
As a result, 
{\em Automated Machine Learning}, or  AutoML for short, has become an increasingly popular research topic over the past years
\cite{Matthias2015}.  
}

\cm{
Despite all the advances in machine learning and AI techniques and availability of impressive computing power, the data science process takes a long time as the process often requires manual effort from data scientists to prepare the data for machine learning. 

We are building a suite of efficient algorithms to automatically extract data types, patterns, context and structural relationships from any type of data to accelerate the data science project timeline. With our solutions, data scientists and domain experts will be provided with automatically-generated, data-driven insights when exploring new data sets. Such insights will help to automatically generate features from raw data. 

Conventional machine learning crucially relies on domain experts to perform a range of machine learning tasks such as 
data exploration, problem formulation, feature extraction, hyper-parameter tuning and model deployment. As the rapid growth of machine learning applications and the complexity of these tasks is often beyond non-experts, there is an increasing demand for developing automated machine learning procedures that can be easily applied without experts knowledge. 
As a result, 
{\em Automated Machine Learning}, or  AutoML for short, has become an increasingly popular research topic over the past years
\cite{Matthias2015}.  
}

\cm{
as the availability of big data becomes widespread and the human machine learning experts remain scarce \cite{Matthias2015}.
}

\cm{
In AutoML, data profiling, the systematic analysis of the content of a dataset, is one of the most difficult and time-consuming parts due to lack of understanding of the data. Data from real world problems can be very dirty. Incompleteness, noise and inconsistency may comprise data analytics accuracy. Conventional data profiling methods are knowledge-based or rely heavily on domain experts \cite{chu:2013}\cite{Chu2015} (\textcolor{blue}{I reworded this paragraph, please check if these references still apply.}). We aim to design automated data profiling methods that allow for a fast and accurate understanding of the characteristics of a given dataset, such as data ID fields, statistical information of data columns, and dependencies between data columns. 
}



\cm{
In this paper, 
we aim to design automated data profiling methods for tabular data that allow for a fast and accurate understanding of its structure  dependencies between data columns. Our work falls into the area of {\em Automated Machine Learning} \cite{Matthias2015}, with the goal to develop automated machine learning procedures that can be easily applied without expert knowledge. 
}


The key to our structural relationship discovery is efficient extraction of important {\em functional dependency} between columns of a data table $\mathcal{T}$ with $N$ columns. In simple words, a column combination $C$ {\em functionally determines} a column $Y$, or $C\rightarrow Y$, if and only if
 each set of $C$ values is associated with precisely one $Y$ value. 
 The dependency is {\em minimal} if it no longer holds after removal of any columns in $C$. Existing approaches
 on uncovering functional dependencies 
 focus on extracting {\em all} solutions of $C$ for a given column $Y$
 in a data table.
 Unlike these methods, our approach
recursively extracts the most important structural relations
for the purpose of data profiling and understanding.
Specifically, we discover those minimal functional dependencies that tend to contain fewer columns.

Another key to our approach is the recursive strategy. We start  
the discovery from the combination of all columns, i.e., the row index, as the initial $Y$. We find descendants of $Y$ defined as those columns that are functionally dependent on $Y$. From its descendants, we attempt to  find a small set of columns $C$ such that 
$C\rightarrow Y$. If such $C$ exists, we then recursively apply the same process to each column (or combination of columns) in $C$ until failure. 
The outcome of the whole procedure is a skeleton of a schema tree with nodes representing columns and a split representing a minimal functional dependency between a parent and its children.
After the skeleton is extracted, we then attach the remaining columns not in the skeleton to one of the skeleton nodes as descendants to complete the tree construction. With this schema tree, a data table is decomposed into layers of smaller tables, thus 
revealing the deep dependencies within.
Our main contributions include:



\cm{
The main challenge of the ID fields identification is computational complexity. The naive Brute force search (count $r_A$ for any subset $A\subset \{1,2,...,p\}$ until $r_A = n$) is time consuming with exponential computational complexity $O(2^p)$. To address this issue, we proposed an efficient algorithm. In addition, the applications of ID field identification need to be explored and validated using real-world problems. 
}
\begin{itemize}
\item A forward addition (FA) algorithm that identifies a solution of column combination $C$ that
functionally determines a given column $Y$ with only $O(\log N D)$ distinct count evaluations for a size constraint $D$, i.e., the number of columns in solution $C$ is at most $D$.  
\item Success probability analysis for finding a solution $C$ with size constraint $D$ in one run of FA algorithm.
\item Algorithms for finding the best solutions from multiple runs of FA with size and error constraints.
\item A recursive process to build schema tree giving concise representations of data structure and dependency. 
\item An example of how to utilize the discovered schema tree for feature engineering.
\end{itemize}

Before we proceed, we discuss how our work here on structural relationship discovery based on functional dependency is related to the widely studied statistical (or probabilistic) dependency between columns. Notice that
functional dependency between columns is {\em deterministic} while the 
statistical dependency is {\em probabilistic}. In this view, statistical dependency can be viewed as a generalization of the functional dependency. 
Many methods has been developed to detect strong statistical dependencies between columns in a data table, from using simple metrics such as Pearson correlation \cite{pearson1895note}, mutual information \cite{journals/bstj/Shannon48} to more complex methods such as graphical models \cite{koller2009probabilistic}.  
 

\section{Background Review} \label{sec:backgroud}
In this section, we first provide necessary background and then discuss related work. We start with notations and definitions. Let ${\mathcal T}$ be a data table with $N$ columns and $R$ rows. For the work presented in the paper, we assume there is no duplicated rows in ${\mathcal T}$.
 Let $C$ and $Y$ be two sets of column combinations.
\begin{definition}
 For a column combination $C$, define
$r(C) = \mbox{distinct row count with columns $C$ from table} \ {\mathcal T} $, 
and $|C|$ as the number of elements (columns) in $C$. 
Note that $r(C)$ is non-decreasing with respect to column additions to $C$, and its maximum value is $R$.
\end{definition}

\begin{definition}
\label{def:minimal-dependence}
$C$ functionally determines $Y$, or $C\rightarrow Y$, if and only if  $r(C) = r(C\cup Y)$, 
i.e.,  each $C$ value is associated with precisely one $Y$ value. It is 
a \textbf{minimal functional dependency} if removal of any column from $C$ breaks the dependency.
\end{definition}
\begin{definition}
Define $Descendant(C)$ as the set of all columns $Y$ such that $C\rightarrow Y$. 
 \end{definition}

\cm{
\begin{definition}
\label{def:minimal-dependence-epsilon}
$C$ $\epsilon$-functionally determines $Y$, or $C\xrightarrow{\epsilon} Y$, iff  $r(C\cup Y)(1-\epsilon) \leq r(C)$ (i.e., $error(C,Y)\leq \epsilon$). It is 
a \textbf{$\epsilon$-minimal functional dependency} if removal of any column from $C$ breaks the dependency.

\end{definition}
}

\begin{definition}
\label{def:minimal-unique}
$C$ is called \textbf{minimal unique} if 
$r(C)=R$, and if removal of any column from $C$ breaks the equality. 
Notice this is a special case  of Definition 2  when $Y$ is the set of all columns.
\end{definition}

\cm{
\begin{definition}
A functional dependency $X\rightarrow Y$ is a minimal functional dependency if removal of any attribute $A$ from $X$ means that the dependency does not hold any more; that is, for any attribute $A \in X, (X – \{A\})$ does not functionally determine Y. A functional dependency $X \rightarrow Y$ is a partial dependency if some attribute $A \in X$ can be removed from $X$ and the dependency still holds; that is, for some $A \in  X, (X – \{ A\}) \rightarrow Y$.
\end{definition}

\begin{definition}
 ((Non-)Unique Column Combination). A column
combination K ⊆ S is a unique for R, iff ∀ri
,rj ∈ R, i , j : ri[K] ,
rj[K]. All other column combinations are non-unique.
\end{definition}

relation R with schema S (the set of n attributes) a unique column combination is a set of one or more attributes whose projection has only unique rows. In turn, a non-unique column combination has at least one duplicate row

 A functional dependency $X\rightarrow Y$ is a minimal functional dependency if removal of any attribute $A$ from $X$ means that the dependency does not hold any more; that is, for any attribute $A \in X, (X – \{A\})$ does not functionally determine Y. A functional dependency $X \rightarrow Y$ is a partial dependency if some attribute $A \in X$ can be removed from $X$ and the dependency still holds; that is, for some $A \in  X, (X – \{ A\}) \rightarrow Y$.
 }
 
 
 
For a column combination $C$, it is easy to obtain
$Descendant(C)$ by simply checking the equality $r(C)=r(C\cup Y)$ for any column $Y$. The inverse problem is much harder: for a given $Y$, find $C$ such that $C\rightarrow Y$ is a minimal functional dependency. Many algorithms have been proposed in the literature, which can be classified into 'column-wise' algorithms (\cite{Huhtala:1999,he2013}), 'row-wise' algorithms (\cite{lopes2000efficient,wyss2001fastfds,y2006}), or hybrid methods (\cite{Abedjan:2011,papenbrock2017}). 
These methods focus on finding {\em all} solution sets of $C$ since finding one solution is considered as a simple problem. 
For example,
a commonly used approach for finding one solution of minimal unique  is what we call the 'Backward Elimination' (BE) algorithm. It starts from the complete set (all columns), then it recursively attempts to eliminate a column from the set to maintain the distinct count $R$. A solution is obtained when such operation is no longer possible. It is however a much harder problem to find {\em all} solution sets of minimal functional dependency, 
as in the worst case, the number of possible solutions is  exponential in $N$. Therefore, earlier work seek algorithms that can find all solutions with computational complexities that are polynomial with respect to the size of the solution set. In contrast to these methods, we do not derive all solutions of $C$ that satisfy the functional dependency but instead focus on finding important functional relations, especially those $C$ that contain fewer number of columns. By extracting these functional dependencies recursively, they can be succinctly expressed in the form of a tree that can be very useful for data understanding and greatly facilitate downstream machine learning. 


The rest of the paper is organized as follows. In Section 3, we present our FA algorithm to find one solution for
near functional dependency and show that it favors  short solutions with random column permutations. Section 4 presents algorithms to find  the best solutions with multiple runs of the FA algorithm. In Section 5,  we present an algorithm to build a scheme tree for a data table using recursive functional dependency discovery. In Section 6, an application of the discovered schema tree for feature engineering is demonstrated. We conclude in Section 7.

 

\cm{
Limited research has been done to identify ID fields. As mentioned, a brute-force approach requires $O(2^p)$ computational complexity. To reduce the complexity, Lacroix et al.~\cite{Lacroix:2003} limited the search space to only a small number of adjacent columns. By doing so, their approach may only be able to find local ID fields but omit the ID fields that are composed of columns far from each other in the data table. In database area, the ID fields is considered as composite keys. Different methods were proposed to find the composite keys
\cite{Huhtala:1999,Gunopulos:2003,Abedjan:2011,y2006,zhang2010,Heise:2013}. For example, Sismanis et al.~\cite{y2006} proposed the \textit{Gordian} method to find all composite keys in a dataset. They formulated 
the problem as a cube computation that corresponds to the computation of the entity counts of all possible column projections. This method involves multiple procedures and is unnecessarily complex with high computational complexity (Theorem 1 in \cite{y2006}). In addition, their method is restricted to datasets that satisfying their additional assumptions. In mathematics, $r_C$ can be viewed as a \textit{submodular function} \cite{Heise:2013}. The submodular set function optimization is to find a set $C$ that maximize $r_C$. However, the assumptions of the submodular set function optimization problem are quite general. There is no optimal solution but an approximate solution for the submodular optimization problem, which is not a satisfactory solution for ID fields identification.

On the other hand, traditional methods for insights extraction mainly rely on text analytics according to the context and contents in the corresponding columns~\cite{Aggarwal}. These methods mostly focus on the information within one single column and fail to establish associations to other columns. The ID fields can provide valuable correlation information among the columns. Hence combing the ID fields with text analytics, we are able to extract more accurate insights from datasets.
}
\begin{figure*}[htpb]
\centering
\includegraphics[width=5.3in, angle=270]{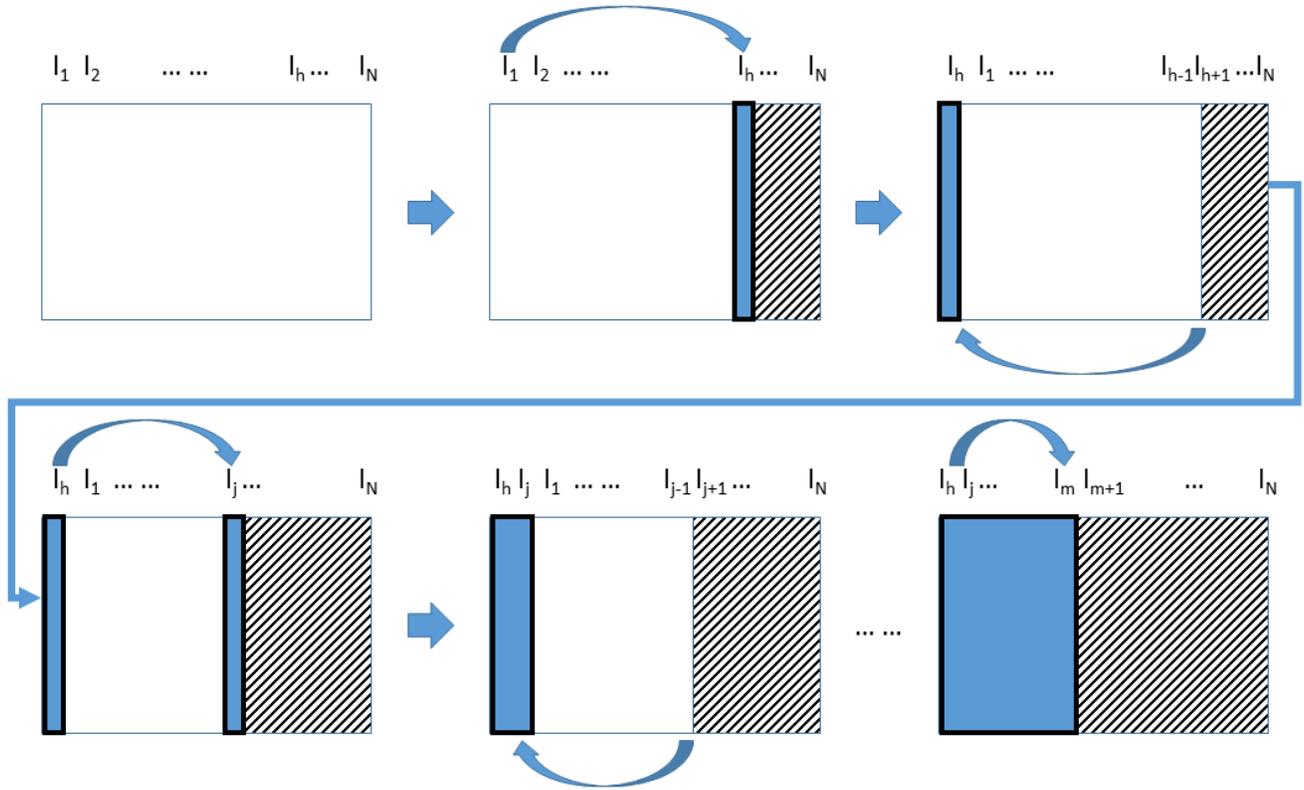}
\caption{Illustration of Algorithm 1.  The rectangle represents the data matrix with columns $l_1,\ldots,l_N$. Panel 1: data matrix with columns $l_1,\ldots,l_N$; Panel 2 illustrates the result of the first iteration (line 6 with $i=1$), where column $l_h$ (blue) is the first selected column. Columns in the shaded area are then excluded from further consideration. Panel 3 shows column $l_h$ is moved to front with the rest remaining in the same order (line 7). Panel 4 and 5  shows the result of the second iteration. Panel 6 shows the final solution set consisting of the blue columns.}
\label{fsa}
\vspace*{-0.2in}
\end{figure*}

 
 
\section{Forward Addition (FA) Algorithm for One Solution}
In this section, we present a lightweight forward addition (FA) algorithm to obtain a single solution of $C$ such that $C\rightarrow Y$ for a column combination $Y$ with a size constraint $D$, i.e., $|C|\leq D$.
We first illustrate our method for the case of finding minimal uniques and show
FA is much faster for a table with large number of columns $N$ comparing to  BE Algorithm (Section 2). 
We discuss the probabilistic version of FA using random column permutations and show that this probabilistic variation favors shorter solutions of $C$ (i.e. those with fewer columns). 
Finally we extend our algorithm to the general case of (approximate) functional dependency. 

\subsection{FA Algorithm}
We first present FA algorithm for deriving one solution of minimal unique $C$ (Definition \ref{def:minimal-unique}), with a size (number of columns) no bigger than $D$.
Given a sequence of column indices $L_0$
which is a permutation of $\{1,2,\ldots,N\}$, 
FA Algorithm attempts to find $C$ in at most $D$ sweeps of $L_0$. We start with an empty set $C$ and in each sweep, one element of $L_0$ is added to $C$ until $r(C) = R$. 
The algorithm is presented both in pictorial (Figure~\ref{fsa}) and pseudo code form (Algorithm~\ref{alg:fsa}). 
A detailed description is as follows.

\begin{algorithm}[htb]
\caption{FA Algorithm for Finding a minimal unique with size $\leq D$ for
Table ${\mathcal T}$ }
\label{alg:fsa}
\cm{
\textbf{Input}: Table ${\cal T}$ and size parameter $D$\\
\textbf{Output}: {a minimal unique $C$ with size less than $D$}\\
}
\begin{flushleft}
 \textbf{ Input}: $L_0$,
 a permutation of $\{1,2,\ldots,N\}$
\begin{algorithmic}[1]
\STATE Initialize:
$C\leftarrow\emptyset$, $d\leftarrow 0$, $L\leftarrow L_0$
\STATE \textbf{while} {$r(C) < R$ and $d<D$} 
    \STATE \hspace{0.15in} $T\leftarrow C$
	\STATE \hspace{0.15in} \textbf{for} {$i=1$ to $|L|$}
		\STATE \hspace{0.3in} $T \leftarrow (T, l_i)$, $l_i \doteq i$th element of $L$
		\STATE \hspace{0.3in} \textbf{if} {$r(T) = R$} \textbf{then}
			\STATE \hspace{0.5in} $C\leftarrow (C, l_i)$, 
			 $L\leftarrow (l_1, l_2,..., l_{i-1})$, $d=d+1$ \\ 
            \STATE \hspace{0.5in} \textbf{break} (go to line 2)
\STATE \textbf{if} {$r(C)=R$} return $C$
       \textbf{else } return $\emptyset$
\end{algorithmic}
\end{flushleft}
\end{algorithm}

First, initialize $C=\emptyset$, $d=0$
and $L=L_0$ (line 1 and the first panel). Next, we create a temporary set of columns $T$ that is identical to $C$ (line 3). We now add indices in $L$ one by one to $T$ until for the first time the distinct row count for column combination $T$ reaches $R$ (line 4-6 and panel 2). Suppose this is the $h$th element in $L$.  From here, we can conclude that $\{l_1,l_2,\ldots,l_h\}$ must contain a solution set of minimal unique, and column $l_h$ is required in the solution, so we add $l_h$
to $C$ and ignore the remaining columns after the first $h$ columns in $L$
(line 7-8 and panel 3 in the figure where the shaded columns are to be discarded). This completes one sweep of $L$. With an updated $C$, we go back to Step 2 and repeat this process at most $D$ times where in each time we find one column in $L$ to add to solution set $C$  until either $r(C) = R$ or we exceed the size limit $D$ (line 2 in algorithm, and 4th and 5th panels of Figure 1 shows the second iteration of this). If $r(C)=R$, then $C$ is a  minimal unique with $|C|\leq D$. Otherwise we failed to find a solution (line 9 or panel 6). 

It is helpful to understand how FA algorithm works in the case of multiple minimal uniques using illustrative examples.
Let the initial $L$ be $L_0=(1,2,\ldots,N)$. Assume first
$\{1,4,6\}$ is the only minimal unique. Then FA algorithm would find column 6 first, then column 4, and finally column 1 to add to $C$ in three sweeps of $L$ (where each time $L$ shrinks by discarding the elements after the selected one (shaded columns in Figure 1)). Now suppose there are two solutions: $C_1=\{1, 4, 6\}$ and $C_2=\{1,5,7\}$. As the largest numbers in $C_1$ and $C_2$ are 6 and 7 respectively, and 6 is less than 7,  FA algorithm would return $C_1$ as the solution since the search reaches $C_1$ first (line 4 and 5). Suppose the two solutions are: $C_1=\{1,4,6\}, C_2=\{2,5,6\}$. In this case
the largest number in each solution set ties (both are 6). Then after 6 is added to $C$,
since $4$ from $C_1$ is less than $5$ from $C_2$, this becomes the tiebreaker with the addition of the 2nd column index, which would return $C_1$ as the solution.  
In general, FA would return with a solution whose largest column index in the permuted table is the smallest among all solutions. If there is a tie, the tie-breaker would be determined by next additions in the same manner.



\subsection{Computational Complexity} 
During each sweep, column indices in $L$ are incrementally added to a  temporary set $T$ until the distinct row count reaches $R$ (line 5-6 of Algorithm 1). Since the distinct row count with respect to column addition is strictly increasing until it hits $R$, if the number of columns $N$ is large,
we can use binary search to locate this column  with
 $O(\log N)$ distinct count evaluations. Therefore, with $D$ sweeps, the total distinct count evaluations is $O(D\log N)$. 
 In comparison, 
  the required distinct count evaluations  for BE algorithm (Section \ref{sec:backgroud})
is $O(N)$ for the case of $D\ll N$, which implies FA is less 
expensive than BE for a table with a large number of columns $N$.


 
\begin{theorem}
One run of Forward Addition (FA) algorithm requires $O(D\log N)$ distinct count evaluations while one run of Backward Elimination (BE) algorithm requires $O(N)$ distinct count evaluations. 
\end{theorem}




\subsection{FA algorithm with Random Column Permutation}
In Algorithm 1 line 1, $L_0$ is a permutation of the original column indices $\{1,2,\ldots,N\}$. 
We analyze the probabilistic properties of  FA algorithm where the index sequence 
$L_0$ is a random permutation of $\{1,2,\ldots,N\}$, especially when there are
multiple competing minimal uniques. 
Suppose there exists at least one minimal uniques with size no bigger than $D$. We also discuss the success probability that one run of the randomized FA algorithm
will find a desired solution.

Suppose a column combination $C$ is a minimal unique with $d$ columns. Given a random permutation $L_0$ of $\{1,2,\ldots,N\}$, let $m_1, m_2, \ldots,$  $m_d$ be the index of each of $C_i$ in $L_0$. For example, suppose $C=(1,4,6)$ is a minimal unique. With a random permutation $L_0$ of $(1,2,\ldots,N)$,
their corresponding indices
in $L_0$ becomes $(5,3,9)$ respectively. This means, the 5th column in $L_0$ is the 1st column in the original table, the 3rd column in $L_0$ is the 4th column in the original table and the 9th column in $L_0$ is the 6th column in the original table.Then $m_1=5, m_2=3, m_3=9$. For $d\ll N$ with a large $N$, it can be shown that $m_i/N, i=1,\ldots,N$ can be approximated by 
independent and identically distributed (i.i.d.) uniform random variables on $[0,1]$. Let $M=\max(m_1,\ldots,m_d)/N$, which is the largest index of the solution set after permutation, divided by $N$. Then $M\approx \max_{i=1}^d U_i$ where $U_i$ are i.i.d. uniform$[0,1]$. It is easy to show that $M_d$ has the following approximate distribution: $P(M\leq x) \approx x^d$. We summarize the result in the following lemma.

\begin{lemma}
Suppose we are given a table $\mathcal{T}$ and a column combination $C$ of $d$ indices, $C=(C_1,C_2,\ldots,C_d)$. 
After a random column permutation of table $\mathcal{T}$, 
let $m_i, i=1, \ldots, d,$ denote the new index of 
$C_i$ in the permuted table. Then as the total number of columns $N\rightarrow \infty$, $m_i/N, i=1,\ldots,d$ approximates independent $Uniform(0,1)$ random variables. Furthermore, let $M=\max_{i=1}^d m_i/N$. Then $P(M\leq x) \approx x^d$ for any $x\in (0,1)$ as $N\rightarrow \infty$.
\label{lem:M}
\end{lemma}

Now suppose there are $K$ minimal uniques, $C_1,C_2,\ldots,C_K,$ with length $d_1, d_2, \ldots, d_K$. Let
\begin{eqnarray*}
M_k \doteq \hspace*{-0.2in}&& \mbox{
largest index of $C_k$ in the randomly column } \\
&& \mbox{permuted table}, k=1, \ldots, K
\end{eqnarray*}
As we show before, if there is no ties in $\{M_k, k=1,\ldots,K\}$, FA will return the solution $C_j$ whose $M_j$ is the smallest. For the special instance of two minimal uniques with length $d_1,d_2$ and no-overlapping elements, simple calculations by integration show that 
\begin{eqnarray}
P(M_1 < M_2) \approx \frac{d_2}{d_1+d_2}.
\end{eqnarray}
Therefore, if $d_1<d_2$, then $P(M_1<M_2)> 1/2$, and the probability goes to 1 as $d_2$ becomes increasingly larger than $d_1$.
This implies that FA would {\em probabilistically} favor the shorter solutions. It is easy to extend the conclusion to the case when $C_1, C_2$ that share common elements by separating out the common elements from the  non-overlapping elements and use similar integration calculations. We have the following lemma.

\begin{lemma}
Let $C_1,C_2$ be two minimal uniques with length $d_1$ and $d_2$ respectively. If $d_1 < d_2$ then 
\[
P(M_1 < M_2) > 1/2,
\]
where $M_1,M_2$ is the maximal index value defined in Lemma \ref{lem:M}. This implies that if $C_1$ and $C_2$ are the only two minimum uniques, FA algorithm with random column permutation would favor $C_1$ as the returned solution as the probability of returning $C_1$ is higher than that of $C_2$.
\end{lemma}

The above lemma can be generalized to multiple solutions, and to the case with arbitrary solution lengths (not necessarily $d\ll N$). However, the extension to the most general case when the solutions sets are overlapping may need careful treatment. 

\begin{theorem} \label{thm:shortest}
FA with random column permutation (i.e., $L_0$ is a random permutation of $\{1,2,\ldots,N\}$ tends to favor the shorter solutions of minimal unique. In particular, if $C_1,\ldots,C_K$ are the $K$ solution sets with length $d_1<d_2\leq d_3 \leq ... \leq d_K \leq D$ that are non-overlapping, then the probability that $C_1$ be the final returned solution is the highest.
\end{theorem}

Suppose there exist solutions with a size no bigger than $D$ out of all $K$ minimal uniques. Let $Succ$ be the event that we  successfully finds a minimum unique with size constraint $D$ after one run of FA with random permutation. 
Then by Theoroem \ref{thm:shortest}, 
$P(Succ=1)$ is at least $1/K$.
Therefore, for a finite $K$ and an arbitrary small $\epsilon>0$, we can find the a max failure times parameter $F$ such that $(1-1/K)^F < \epsilon$. 
If there are $F$ consecutive failures, we can probabilistically declare that there is no minimal uniques with size less or equal to $D$. It is worthwhile to mention that the performance of BE algorithm is identical to FA in terms of success probabilities when the columns are randomly permuted.

\cm{
\begin{theorem}
Suppose there exist solutions with length less than $d$. Then 
the probability that we will succeed in finding a solution using FA with random permutation is greater than 1/2. Furthermore, the probability that we fail to find a solution with size less than $kd, k\geq 1$ is less than $1/(k+1)$, i.e., $P(R = \emptyset) < 1/(k+1)$. With independent runs that result in $F$ failures, the probability is less than $1/(k+1)^F$, i.e.,

\begin{eqnarray*}
& P(R_i=\emptyset, i=1,\dots,F \ | \ \mbox{A solution with length $d$ exists}) \\
 & <  1/(k+1)^F. 
\end{eqnarray*}
\end{theorem}
\begin{proof}
Since each $R_i$ are independent event, we shall just show that
$P(R_1=0)<1/2$. Notice that the event that $R_1=0$ implies that there is a longer solution 
\end{proof}
}



\subsection{Extension to Functional Dependency}
We generalize Algorithm 1 to the case of
functional dependence. The generalization also 
considers {\em approximate} functional dependence defined below. Again let $C,Y$ be two column combinations in Table ${\mathcal T}$. We consider two definitions of error measurement in order to define 
$\epsilon-$ approximate functional dependency: $C\xrightarrow{\epsilon} Y$.
\begin{definition}
Let $e(C\rightarrow Y)$ be the minimal fraction of rows to be removed
for $C\rightarrow Y$ to hold.
Then
$C\xrightarrow{\epsilon} Y$, if $e(C\rightarrow Y) \leq \epsilon$.
\end{definition}

\begin{definition}
Define $\tilde{e}(C\rightarrow Y) = 1-r(C)/r(C\cup Y)$,
then $C\xrightarrow{\epsilon} Y$ if $ \tilde{e}(C\rightarrow Y) \leq \epsilon.$
\end{definition}

It is easy to show the following lemma that states a monotone property of the error measures with column additions. 
\begin{lemma}
\label{lem:cadd}
$e(C\rightarrow Y)=0$ iff $\tilde{e}(C\rightarrow Y)=0$. For a given $Y$, $e(C,Y)$ is non-decreasing with column additions to $C$. $\tilde{e}(C,Y)$ is non-decreasing with column additions to $C$ when columns in the addition are chosen from $Descendant(Y)$. 
\end{lemma}

\cm{
\begin{definition}
\label{def:minimal-dependence-epsilon}
We define $C\xrightarrow{\epsilon} Y$ if and only if
$C\rightarrow Y$ holds after removing $\epsilon$ fraction of rows.  
$\epsilon$-functionally determines $Y$, or $C\xrightarrow{\epsilon} Y$, if and only if  $r(C\cup Y)(1-\epsilon) \leq r(C)$ (i.e., $error(C,Y)\leq \epsilon$). It is 
a \textbf{$\epsilon$-minimal functional dependency} if removal of any column from $C$ breaks the dependency.

\end{definition}
}
We use two error definitions here due to 
computational considerations. Error in Definition 6 is much easier to compute than Definition 5 (which can be computed using algorithms in \cite{Huhtala:1999}). Furthermore, as Lemma 1 indicates, both definitions coincide when the error is 0. Unfortunately, the error in Definition 6 lacks a monotone property in the general case with column additions, unless the columns are from $Descendant(Y)$. However, this special case is indeed what we use to build a schema tree in Section 5. In the following, Algorithm 2 extends FA algorithm to general functional dependence in a straightforward fashion, 
by replacing the check for minimal uniqueness with the check for 
functional dependency (line 2, 6 and 9). The key to this extension is the monotonic property of the errors in Lemma \ref{lem:cadd}. 


\begin{algorithm}[htb]
\caption{FA Algorithm for finding $C=FA(Y, \epsilon, D)$, such that $C\xrightarrow{\epsilon} Y$ with $|C|\leq D$
for a column combination $Y$ 
}
\label{alg:fsa2}
\textbf{Input}: {$L_0$, a sequence of column indices after column permutation} \\
\cm{
\textbf{Output}: $FA(Y, \epsilon, D)$
}
\begin{algorithmic}[1]
\STATE Initialize:
$C\leftarrow C_0$, $d\leftarrow 0$, $L\leftarrow L_0 \setminus C_0$
\STATE \textbf{while} {$e(C\rightarrow Y) > \epsilon$ and
$d<D$}
\STATE \hspace{0.15in} $T\leftarrow C$
	\STATE \hspace{0.15in} \textbf{for} {$i=1$ to $|L|$}
		\STATE \hspace{0.3in} $T \leftarrow (T, l_i)$, $l_i \doteq i$th element of $L$
		\STATE \hspace{0.3in} 
		\textbf{if} {$e(T\rightarrow Y)\leq \epsilon$}
		\textbf{then}
			\STATE \hspace{0.45in} $C\leftarrow (C, l_i)$, 
			 $L\leftarrow (l_1, l_2,..., l_{i-1})$, $d=d+1$ \\ 
            \STATE \hspace{0.45in} \textbf{break} (go to line 2)
\STATE \textbf{if} {$e(C\rightarrow Y)\leq \epsilon$} return $C$
       \textbf{else } return $\emptyset$
\end{algorithmic}
\end{algorithm}

\cm{
We discuss the initialization of $C_0$ in the implementation (line 1).
An alternative to initializing $C_0$ as $\emptyset$ is to find columns that are essential in 
the functional dependency, i.e., the dependency breaks after the removal of any of these columns.
This is often more computationally advantageous
for running the algorithm multiple times as we will do in the next section. This in fact can be done as a re-initialization step after a solution $C$ is found and
we only need to check columns from $C_0$. 
The latter is often more computationally advantageous. 
}


\section{Best Solution with Size and Error Constraints}
For a set of columns $Y$ from a data table $\mathcal{T}$, in this section, we propose algorithms to find best solutions of $C$ such that $C\xrightarrow{\epsilon} Y$ 
with $|C|\leq D$ for a given error $\epsilon>0$ and size $D>0$. 
These algorithms essentially run FA algorithm (Algorithm 2) multiple times where in each iteration we attempt to improve upon previous solution until failure occurs. We consider the following two scenarios.

\subsection{Shortest Solution with Error Bound}
For a given $Y$ and error upper bound $\epsilon \geq 0$, 
Algorithm 3 find a solution $C$ with the fewest number of columns with at most $F$ consecutive failures such that $C\xrightarrow{\epsilon} Y$. 
The algorithm works with an input maximum number of consecutive failures $F$ after which we can declare no better solutions can be found. A call to Algorithm 2 is performed in each iteration (line 4) to attempt to find a solution with the given error bound $\epsilon$ and the current achievable
minimal size $\widetilde{D}$ (initialized to be $D$). After the first solution is found, the algorithm tries to find
a better solution with a smaller size (fewer columns) if successful (line 6 and 7). 
The algorithm terminates if the maximum number of attempts has been reached or the solution size reaches 1 (line 2).


\begin{algorithm}[htb]
\caption{Find $Shortest_1(Y, \epsilon, D)$, the shortest solution of $C$  with $|C|\leq D$, such that 
$C\xrightarrow{\epsilon} Y$ for  a column combination $Y$ 
}
\label{alg:smallest-y-1}
\textbf{Input}: Maximum number of consecutive failures $F$ (e.g. 10)
\begin{algorithmic}[1]
\cm{
\STATE Run FA algorithm at most $K$ times until one solution is found. If no solutions found, then return 'no solution'. Otherwise, continue and refer this solution as $C^{(1)}$. 
\STATE Find all essential columns from $C^{(1)}$. Assign it to $C_0$. 
\STATE For each column in $C^{(1)}$, find its equivalent columns, and remove these from the data. 
}
\STATE Initialize $\widetilde{D} \leftarrow D, k \leftarrow 0, C \leftarrow \emptyset$.
\STATE \textbf{while} {$k < F$ and $\widetilde{D} >1$} 
    \STATE \hspace*{0.15in} $L_0\leftarrow$ random permutation of $(1,2,\ldots, N)$
	\STATE \hspace*{0.15in} $S=FA(Y,\epsilon,\widetilde{D})$ using Alg. 2 with input $L_0$
	\STATE \hspace{0.15in} \textbf {if}  $S=\emptyset$ \textbf{then} $k=k+1$
	\STATE \hspace{0.15in} \textbf {if} 
	$(C=\emptyset$ and $S\neq \emptyset)$ 
	\textbf {or} ($C\neq \emptyset$ and $0<|S|<\widetilde{D}$)
	\STATE \hspace{0.15in} \textbf{then} 
$C\leftarrow S$; $\widetilde{D} \leftarrow |S|$;  $k\leftarrow 0$
		\STATE\hspace*{0.4in} \textbf{break}; go to line 2
\STATE return $C$
\end{algorithmic}
\end{algorithm}
\vspace*{-0.1in}

\subsection{Shortest Solution while Minimizing Error} 
A more difficult situation is to find the shortest solution $C$ such that $C\xrightarrow{\epsilon} Y$ and $|C|\leq D$ {\em while minimizing} $e(C\rightarrow Y)$ for a given $\epsilon$ (i.e. minimizing $e$ first, then the solution size). A naive approach is to run
Algorithm 3 on a grid of $e$ values, $e<\epsilon$, and return the solution with the minimal possible $e$. If multiple solutions exist with the minimal value of $e$, return the shortest one.

Algorithm 4 provides an alternative for finding the minimal $e$ using hill climbing, avoiding the grid search with at most $F$ consecutive failures. In the first iteration,
it finds a candidate solution $S$ by running Algorithm 2 
(line 3-4). Next, it attempts to reduce the error rate $e$ using sequential column additions to $S$ with the max size $D$, i.e., $ColAdd(S, D)$ (line 5,6). This is due to the fact that, from Lemma \ref{lem:cadd}, the error rate can be reduced with column additions. We assign this updated column set as our initial solution $C$ along with its error rate $e$. 
In the following iterations, we perform a similar updating process (line 3, 4, 7-10) in the attempt to achieve an even smaller error rate until we fail consecutively $F$ times. At this point, we have found the smallest reachable $e$. Finally, we find the shortest solution using Algorithm 3 with the given minimum value of $e$ (line 11, 12). 

\begin{algorithm}[htb]
\caption{Find $Shortest_2(Y, \epsilon, D)$, the shortest solution $C$, among all solutions $S$ such that 
$S\xrightarrow{\epsilon} Y$  for  a column combination $Y$, with minimum $e(S\rightarrow Y)$, and $|S|\leq D$.}
\label{alg:smallest-y-epsilon}
\textbf{Input}: Maximum number of  consecutive failures $F$ (e.g. 10)
\begin{algorithmic}[1]
\STATE Initialize 
$e\leftarrow \epsilon, k \leftarrow 0, C \leftarrow \emptyset$, $new\leftarrow 0$
\STATE \textbf{while} {$k < F$ and $e>0$} 
    \STATE \hspace*{0.15in} $L_0\leftarrow$ random permutation of $(1,2,\ldots, N)$
    \STATE \hspace*{0.15in} $S=FA(Y,e,\widetilde{D})$ using Alg. 2 with input $L_0$
    \STATE \hspace{0.15in} \textbf {if} 
	$(C=\emptyset$ and $S\neq \emptyset)$ \textbf{then}
	\STATE \hspace{0.3in} $C\leftarrow ColAdd(S;D)$, $e\leftarrow e(C\rightarrow Y)$;
	$new \leftarrow 1$; $next$ 
	\STATE \hspace{0.15in}  \textbf {if} 
      ($C\neq \emptyset$ and $S\neq \emptyset$) 
	\textbf{then}
     	\STATE \hspace{0.25in} { $S\leftarrow ColAdd(S;D) $}; $e' \leftarrow e(S\rightarrow Y)$ 
     	\STATE \hspace{0.3in} \textbf {if} $e'<e$ \textbf{then} {$C \leftarrow S; e\leftarrow e'; new\leftarrow 1$}
	\STATE \hspace{0.15in} \textbf {if} $new=1$ \textbf{then} ($k\leftarrow 0; new\leftarrow 0$)
	\STATE \hspace*{0.15in}
	\textbf{else} $k=k+1$
	\STATE \textbf{if} $|C|>0$, run Algorithm 3 to update $C$ by $Shortest_1(Y,e,D)$ with $C$ as the initial solution
\STATE return $C$
\end{algorithmic}
\end{algorithm}

\cm{
\subsubsection{Initialize $C_0$}
We start with a random permutation of the columns first, and proceed to use FA algorithm to find a single solution $C$. Let $C_k$ be the $k$th element in $C$, then we first check whether $C_k$ is essential or not defined as follows.
\begin{definition}
A column $c$ is essential in the minimal function dependency solution sets for a column combination $Y$, iff all solution sets contains $C$.
\end{definition}
\label{def:essential}
A necessary and sufficient condition for a single column $c$ to be sufficient is no solution set can be found after removal of column $c$.

Once we obtain a single solution $C$, we proceed to check whether any of the columns in $C$ are essential according to Definition~\ref{def:essential}. Let $C_0$ be any columns that are found to be essential, we shall start FA algorithm (Algorithm 1 and 2) with the initialization $C=C_0$ instead of empty set (Step 1). Notice that we can actually perform the check before obtaining the 1st solution, but this tends to be more computational expensive as we need to go through all the columns not just those in $C$.

After we obtain one solution set $C$, we test whether each of these are essential. Let $C_0$ be the set of all essential columns in $C$, then we need to always start searching for new solutions with $C_0$ included. If all columns in $C_0$ are essential, then we are done.

Re-initialization of $C_0$ can be performed either after we obtain an initial solution set $C$ (Step 6 in Algorithm 3 and 4), or from the beginning.  }


\subsection{Discussion}
In this section, we proposed algorithms to find the best solution of a column combination $C$ that functionally determines $Y$ for given size and error constraints, using multiple iterations of FA algorithm with random column permutations. 
Our algorithms avoid the exhaustive search of all solutions and utilize the fact that FA algorithm favors a short solution (Theorem 2). 
We want to point out that 
if the goal is to find all solutions of $C$ that functionally determines $Y$, then 
running many iterations of FA algorithm with independent column permutations is not efficient. This is because the number of all permutation is $N!$ for a table with 
$N$ columns, a number  much larger than the number of all subsets, $2^N$ (the maximal number of solutions of $C$). That said, we believe more efficient algorithms can be developed for finding multiple solutions of $C$ by 
adaptively choosing
column permutations $L_0$ based on previous found solutions.
For example, if $C=\{1,2,3\}$ is a first solution found by one run of FA, then in the next iteration, we
can place $\{1,2,3\}$ last in $L_0$ and randomly reorder the remaining columns.  
This way, we are guaranteed to find a different solution if it exists
(since the indices of the 2nd solution will appear earlier).
We leave this to future work.

\cm{
We discuss a computational speed-up in the implementation.
In both Algorithms 3 and 4, $C_0$ (in line 3) is initialized as $\emptyset$ by default. 
We can actually initialize $C_0$
by finding columns that are essential in 
the functional dependency, i.e., the dependency breaks after the removal of any of these columns.
We can perform this initialization right from the start, or reinitialize $C_0$ right after we obtain the first solution (with initial $C_0=\emptyset$) by checking only those columns in the first solution set. 
The latter is often more computationally advantageous. 
}


\begin{table*}[htpb]
\caption{Sample Rows from a Simulated Table of Customer Purchase Orders}
\label{tbl:order}
\vspace{-3mm}
{\fontsize{0.27cm}{0.35cm}\selectfont 
\centering
\begin{tabular}{c|c|c|c|c|c|c|c|c|c|c|c}
 \hline\hline
    orderID & productID & customerID & time & ordertype & fullname & phoneno  & age  & ptype & price & weight & shippingcost \\ \hline 
    1 & 1 & 4 & day 1 & web & Alex Smith & 732-906-9882 & 29 & book & 15 & 0.5 & 4.5 \\
    1 & 6 & 4 & day 1 & web & Alex Smith & 732-906-9882 & 29 & music & 11 & NA & NA \\
    2 & 3 & 2 & day 1 & web & Emma Miller & 908-654-3213 & 36 & beauty & 39& 0.2 & 4.2\\
    3 & 2 & 2 & day 2 & web &  Emma Miller & 908-654-3213 & 36 & clothing & 32 & 0.6 & 4.6\\
    3 & 4 & 2 & day 2 & web & Emma Miller & 908-654-3213 & 36 & games & 18 & NA & NA\\
    3 & 6 & 2 & day 2 & web & Emma Miller & 908-654-3213 & 36 & music & 11 & NA & NA\\ 
    4 &  5 & 3 & day 2 & phone & Kim Dole & 
    973-211-1245 & 45 & grocery & 25& 1 & 2.3\\
    4 &  5 & 3 & day 2 & phone & Kim Dole & 
    973-211-1245 & 45 & grocery & 12 & 1 & 2.3\\
\hline\hline
\end{tabular}
\par}
\vspace{-2mm}
\end{table*}

\cm{
\begin{figure*}[htpb]
\centering
\includegraphics[width=7in]{figures/paperOrderfull}
\caption{
Sample rows from a  table on customers purchase orders on a e-commerce website.
}
\label{tab:example}
\end{figure*}
}

\section{Building a Schema Tree}
In this section, we present an algorithm that builds a schema tree for a data table ${\mathcal T}$ using {\em recursive}
 functional dependency discovery.
The schema tree 
represents the structure dependency between table columns in a simple hierarchical form.
As a result, it is possible to decompose a big data table into
layers of smaller tables. In the next section, we show how this information is utilized 
for automated feature engineering.

\begin{algorithm}[htb]
\caption{Build Schema Skeleton for Table $\mathcal{T}$}
\label{alg:schema}
\begin{flushleft}
\textbf{Input}: size parameter $D$, error parameter $\epsilon$\\
\textbf{Output}: Schema tree skeleton $V$
\begin{algorithmic}[1]
\STATE Define $RecusiveSplit(Y) = \{$
\STATE \hspace*{0.15in} \textbf{if} $Y=\emptyset$ \textbf{then} return $V$
\STATE \hspace*{0.15in} $E \leftarrow Descendant(Y)$. 
\STATE \hspace*{0.15in} $S\leftarrow shortest(Y, \epsilon,D)$,  $S\subset E$, using Algorithm 3 or 4.
\STATE \hspace*{0.15in} \textbf{if} $S=\emptyset$ \textbf{then} return $V$
\STATE \hspace*{0.18in}\textbf{else} grow $V$ by adding child nodes $S$ under node $Y$
\STATE \hspace*{0.4in} \textbf{for each} node $W \subset S$, 
$RecursiveSplit(W)$ 
\STATE  \}
\STATE Add an index (rowid) column $Y$. Set $Y$ as the root node of $V$.
\STATE $RecusiveSplit(Y)$
\end{algorithmic}
\end{flushleft}
\end{algorithm}

Our schema tree is built via a two-step process. In the first step, we build the tree skeleton using Algorithm 5.
Given a table ${\mathcal T}$, we first add a hypothetical {\em rowid} column as the row index  and start building the tree using this as the root node (line 9). Let $Y$ be the current node of interest. Then we attempt to split $Y$ using columns from the set $E=Descendant(Y)$ (Definition 3 and line 3). A split is found if 
there exists a solution $S \subset E$ such that $S\xrightarrow{\epsilon} Y$ and $|S|\leq D$ using Algorithm 3 or 4 (depending on the application context), for the given size constraint $D$ and error constraint $\epsilon$ (line 4).
\cm{
by first finding all its descendant, where a descendant is defined as a column that is functionally dependent on $Y$ (line 1 that defines $Descendant(Y)$). Then we find a column combination $S$ from $Descendant(Y)$ such that is $Y$ has a minimal-functional dependency, i.e., $S\xrightarrow[minimal]{e} Y$, with size constraint $D$ and error constraint $\epsilon$ (line 4 and 5 in the definition of $RecurisveSplit(Y)$). 
}
If a solution $S$ is found, then we split
$Y$ by nodes in $S$ as its children. 
Such a splitting process is done recursively for all nodes in $S$ (line 6,7) until we are no longer able to split (line 5). Note when $\epsilon=0$, such a split indicates a bi-directional functional dependency (or equivalence) between $Y$ and $S$ since $S$ is derived from $Descendant(Y)$. It is important to note that for well designed data tables, the skeleton nodes are usually not decimal valued numeric columns. Therefore we do not include these columns when building the tree skeleton.

Once the schema skeleton is built, the second step is to add the remaining columns to complete the tree construction. For each of the remaining nodes, we attach it to the deepest skeleton node that it functionally depends. As explained in Section 2, this can be done by a simple row distinct count check.  In the following, we shall elaborate our process using simulated and real data.


\subsection{Illustrative Examples} \label{sec:example}
Table \ref{tbl:order} displays sample rows of a simulated table that stores information related to customer purchase orders on an e-commerce website: an order is placed by a customer at a specific time, and each order can contain multiple products spanning multiple rows. 
The table also contains the associated customer and product information. Here proper column names are shown with the contextual information so that a domain expert can easily understand the structure relations between columns. Understanding column relations will be much more difficult if this contextual information is removed.
\begin{figure}[htpb]
\centering
\includegraphics[width=3in]{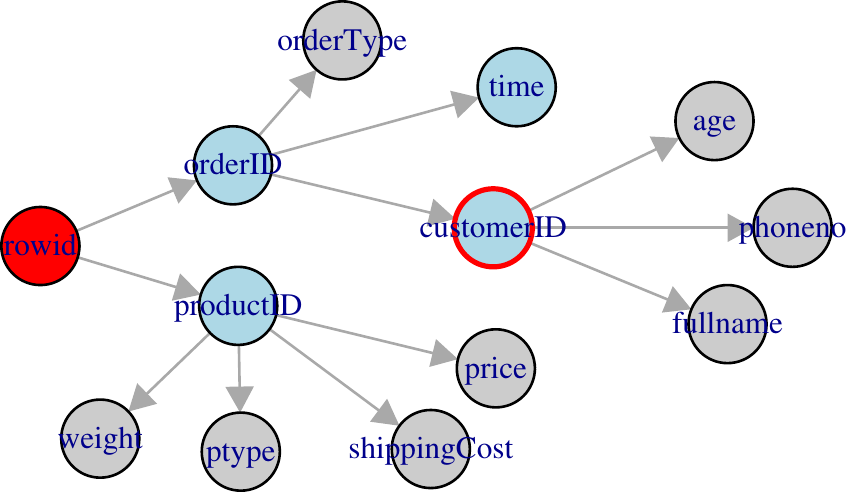}
\caption{Schema tree of Order Table 1.}
\label{fig:orderSchema}
\end{figure}

Figure~\ref{fig:orderSchema} shows the discovered schema tree with $\epsilon=0$ and $D=3$,
where the red node is the root node representing the row index, blue nodes indicate
the discovered tree skeleton using Algorithm 5, and the gray nodes are the rest of the leaf nodes attached in the second step of tree construction. 
Let $\leftrightarrow$ indicate
the {\em joint} bidirectional functional dependency, and  as before, let 
$\rightarrow$ indicate functional dependency.  Figure~\ref{fig:orderSchema} can be interpreted as follows: 
\begin{eqnarray*}
rowID & \leftrightarrow & (orderID,
\ productID) \nonumber\\
orderID & \leftrightarrow & (customerID, \ time) \nonumber\\
orderID & \rightarrow & orderType; \\ customerID & \rightarrow & phoneno; \ \ 
productID \ \rightarrow \ price 
\nonumber
\end{eqnarray*}
where the first two rows are regarding to the blue nodes, and the next two are sample relations concerning the gray nodes. 
The tree skeleton (red and blue nodes) is built using Algorithm 5 while removing {\em weight} and {\em shippingCost} as these are decimal valued columns and not initially included (decimal valued numeric columns can be auto-detected by methods in \cite{huseyin.sdd}). In the second step, we attach each of the remaining columns (gray nodes) to the deepest blue (or red) node that it functionally depends.
With this schema tree, we can actually extract several smaller tables:
a 'Product' table from {\em productID} and its children, a 'Customer' table from {\em CusomterID} and its children, and a smaller 'Order' table from {\em orderID} and its children. Our big table is, in fact, the merge of these three tables.

The second dataset contains broadband home router data records of customers from a  network carrier during a 30-day period.  It consists of $27$ columns and $238330$ rows, where columns are device' ID, type, associated network nodes and types, the customer information, and time series of several KPIs. 
\begin{figure*}[htpb]
\centering
\includegraphics[width=5in]{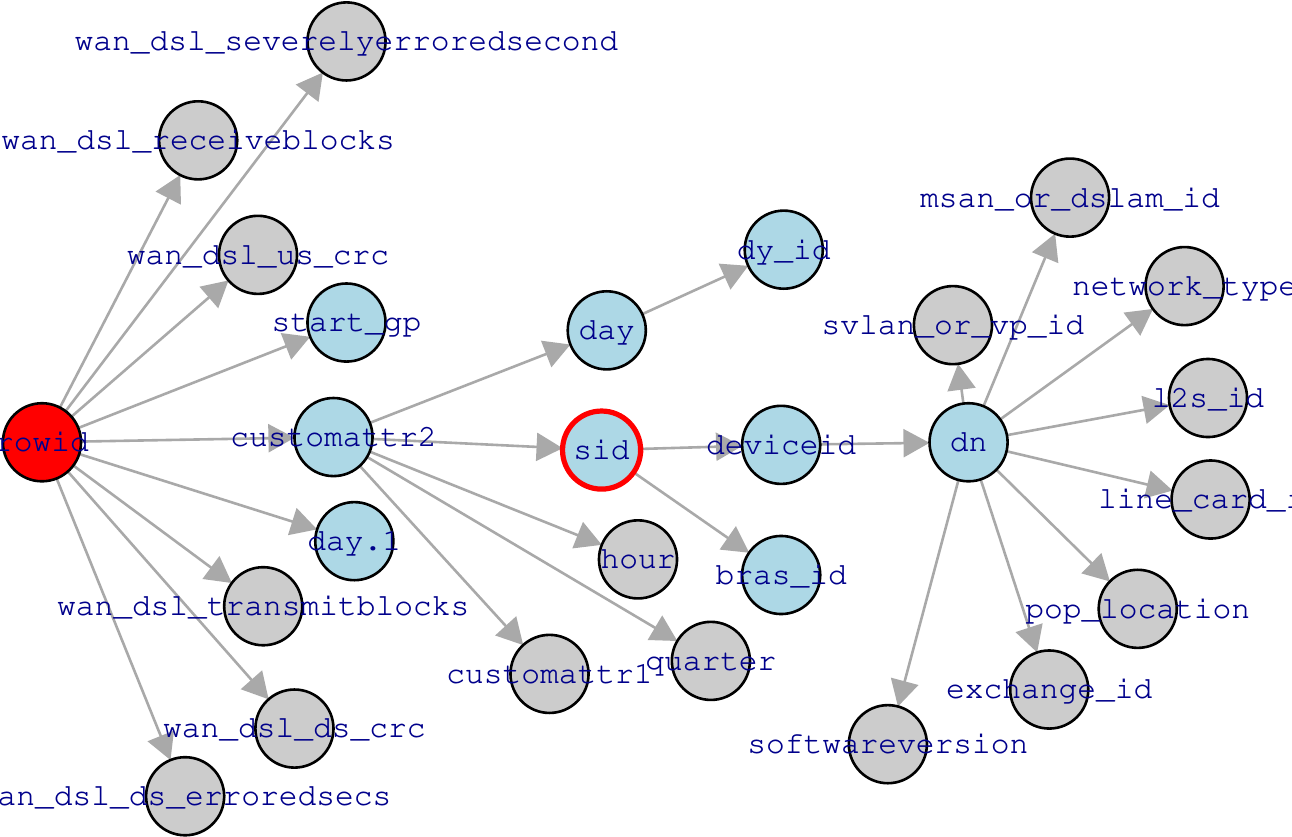}
\caption{Schema Tree of Home Router Device Data}
\label{fig:BTschema}
\end{figure*}

 Figure~\ref{fig:BTschema} shows the discovered schema tree with $\epsilon=0$ and $D=3$. The first split from the red root node representing row index is $rowid \leftrightarrow (start\_gp, $ $ customerattr2,day.1)$, and the next split is $customerattr2 \leftrightarrow (day,$ $sid)$. Furthermore, $day\leftrightarrow day\_{id}$, 
$sid \leftrightarrow (deviceid, bras\_{id})$,
$dn\leftrightarrow deviceid \leftrightarrow dn$. The leftmost 6 gray nodes are associated KPI time series
(it is time related by noticing
the {\em day.1} of the 1st split from the root), 
and the rightmost gray nodes are device-related attributes. Again, with the schema tree, we can decompose the table into layers of 'Device' table from node {\em deviceid}, 
a 'customerattr2' table from node {\em customerattr2}, and a KPI time series table for each {\em customerattr2/ start\_{gp}} combination. Understanding the column relations is much more straightforward with our diagram.

\begin{figure*}[htpb]
\centering
\includegraphics[width=3.2in]{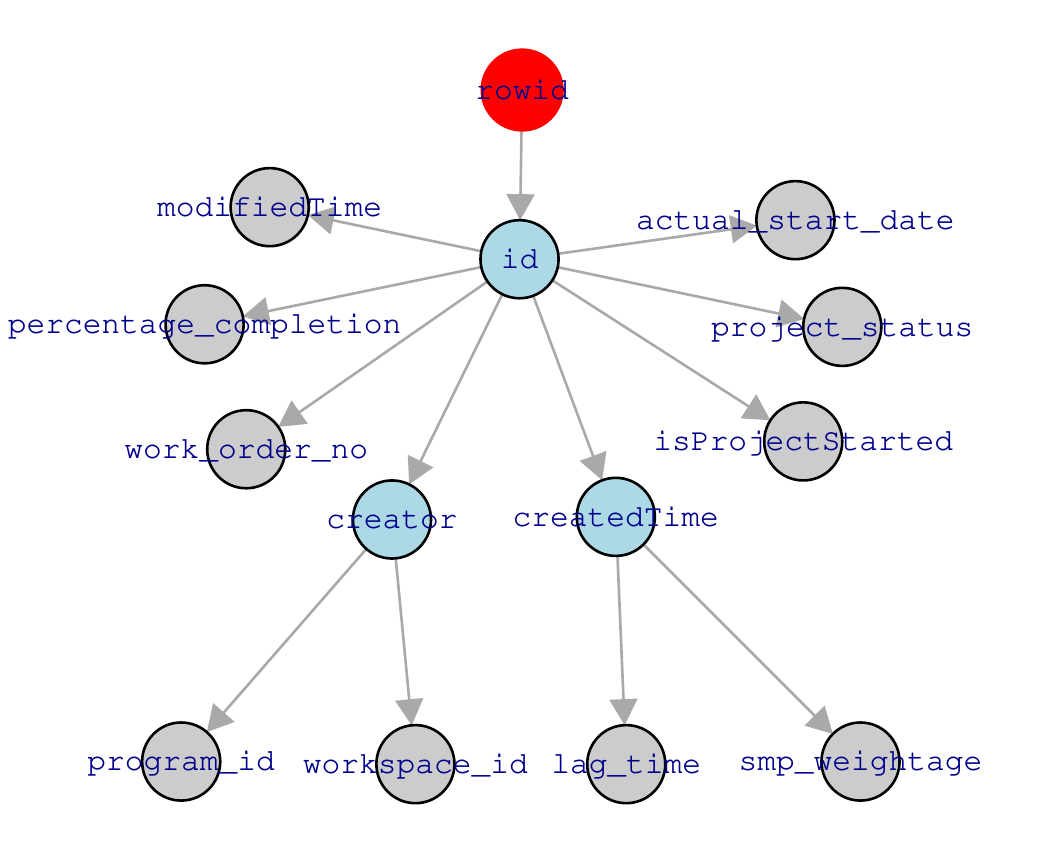}
\includegraphics[width=3.2in]{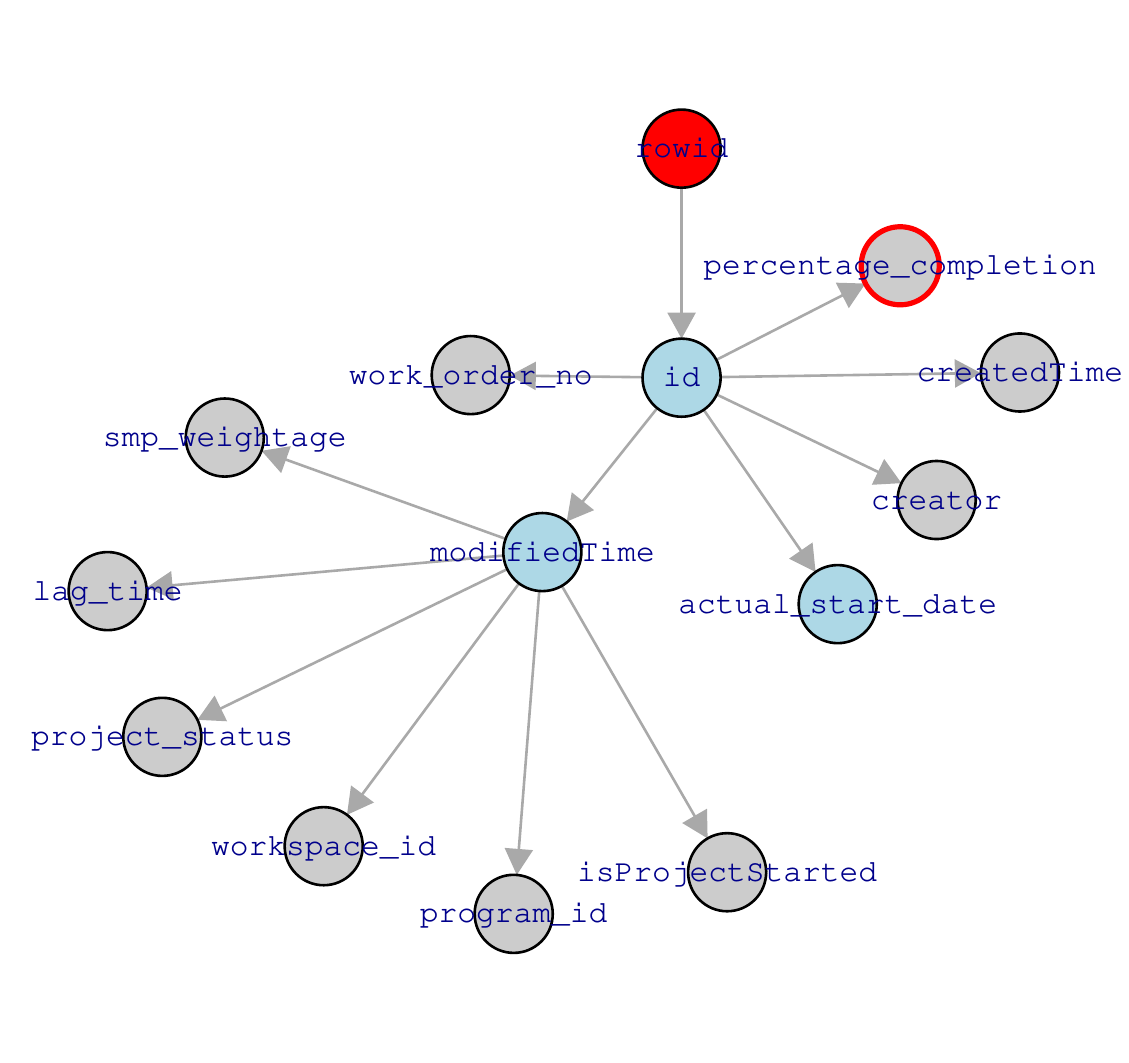}
\caption{Two Realizations of Schema Tree for Project Data Table}
\label{fig:SMP}
\end{figure*}

Finally, we want to comment that the schema tree derived from our algorithm may not be unique, especially for tables with complex column relationship. For example, if there are multiple short functional dependencies with the same length, the algorithm will randomly pick one of these to follow through. Figure~\ref{fig:SMP} shows two different realizations of our algorithm for an internal data table named 'Project' describing the progress of ongoing projects. From the root node $rowid$, both scheme trees find a single column $id$ as the 1st level descendant. However, the trees differ in the next blue node split: the tree on the left is split  by $(creator, createdTime)$, and the tree on the right is split  by ($modifiedTime, actual\_start\_date$). Although both solutions are valid, one may prefer the first scheme tree. This is because $(creator, createdTime)$ is the combination of a label and DateTime column, but 
$(modifiedTime, actual\_start\_date)$ are two DateTime columns, and the formal is a more natural table key than the latter.  In this case, we are using the contextual information in the columns to select a more appropriate solution. How to  automatically infer the contextual information of columns and utilize this in the selection of schema tree is part of our future work.

\section{Applications to Auto Feature Engineering}

Our schema tree built automatically from a data table not only provides valuable information for the data scientists to better understand the table, it can also directly benefit the downstream machine learning tasks. In this section, we show how 
information from the schema tree can be used for {\em automated feature engineering}. 
Traditionally, feature engineering often requires handcrafting  based on contextual understanding. 

We shall illustrate this using the simulated order Table 1.
Suppose 'CustomerID' is an anchor node 'customer' and we are interested in engineering features for each customer. 
First, we notice that  there are two types of relationship from {\em CustomerID} to any other node: \textbf{one-to-one} and \textbf{one-to-many}. In the {one-to-one} relation, each {\em CustomerID} determines a single instance of the other variable, which in fact implies that these variables are either functional equivalent to or descendants of {\em CustomerID} in the schema tree. The remaining  nodes in the schema tree have a  
{one-to-many} relation from {\em CustomerID} which imply that each {\em CustomerID} corresponds to multiple instances of the variable.

In this example, {\em CustomerID} has three gray nodes as descendants: {\em age} , {\em fullname}, {\em phoneno}. 
This implies that for each of the three nodes, there is a \textbf{one-to-one} relation from 
{\em CustomerID}, so we can take the instance value {\em directly} as a feature for each {\em CustomerID}. On the other hand, for each of the remaining columns, there is a \textbf{one-to-many} relation as
each distinct {\em CustomerID} will have one or several instances of its values. 
(Take the node {\em Price} as an example. In Table 1, Customer 4 have 2 price instances, Customer 2 have 4 price instances.)
Therefore we need to aggregate these values to create a uniform number of  features for each {\em CustomerID}. 
In this case, for any node ${\mathcal N}$, we first find the shortest path from {\em CustomerID} to ${\mathcal N}$ in the schema tree, which gives us a bottom-to-top aggregation path towards {\em CustomerID}. 
For example, the path from {\em Price} to {\em CustomerID} while ommiting {\em rowid} as it is a hypothetical node is:
$
Price \rightarrow ProductID \rightarrow OrderID \rightarrow  CustomerID. 
$
We further omit the {\em one-to-one} relations in the path as no aggregation is necessary, 
and so the aggregation path is 
simplified to 
$
Price \Rightarrow OrderID \Rightarrow CustomerID, 
$
where $\Rightarrow$ indicates a many-to-one relation.
Finally, we extract all the subpaths and use a set of pre-specified aggregation functions to generate features. Specifically, {\em Price} is aggregated by
\cm{
\begin{eqnarray} 
Price & \xRightarrow{f_1} & OrderID \ \xRightarrow{f_2} \  CustomerID, \\
Price & \xRightarrow{f_3} &  CustomerID
\end{eqnarray}
}
\begin{equation} 
\resizebox{.98\hsize}{!}{$Price  \xRightarrow{f_1}  OrderID \xRightarrow{f_2}   CustomerID, \ \ \ 
Price  \xRightarrow{f_3}   CustomerID$}
\nonumber
\end{equation}
where $f_1, f_2, f_3$ are pre-specified aggregation functions depending on the characteristics of the variable to be aggregated. For numeric variables, these functions can be 
max, mean, standard deviation, interval probabilities or quantiles.
For categorical variables, it can be distinct counts. This way we can come up with a set of features such as: the maximum ($f_2$) of average ($f_1$) price per {\em orderID}, or the standard deviation ($f_3$) of the price. This approach is related to recent
research on automated feature engineering \cite{DFS-MIT,OneBM}. 
However, these methods
focuses on  relational databases with a known schema with an assumption that each table has a shallow structure. In contrast, our method proposed here do not assume prior schema knowledge and can work with data tables with layered structure. In addition, our methods can also be extended relational databases by joining schema graphs together (not illustrated here). 

\section{Conclusion}
In this paper, we developed an
automatically-generated, data-driven schema tree to represent the structural relations between columns of a data table that can greatly facilitate a data scientist when exploring a new dataset. The key to our approach is the recursive extraction of important {\em functional dependency} between columns of the data table,  where we proposed a forward addition algorithm that requires much less computation compared to existing approaches.



\appendix

%
\bibliographystyle{IEEEtran}
\bibliography{main}

%

\end{document}